\pdfoutput=1

\documentclass[11pt]{article}

\usepackage{url}

\usepackage{graphicx}

\usepackage{authblk}

\usepackage{subfig}

\usepackage[noend]{algorithmic}
\usepackage{algorithm}

\usepackage{amsmath}
\usepackage{amssymb}
\usepackage{amsthm}
\theoremstyle{definition}
\newtheorem{definition}{Definition}
\theoremstyle{plain}
\newtheorem{prop}{Proposition}
\newtheorem{theorem}{Theorem}

\usepackage[round]{natbib}
\usepackage{hyperref}

\newcommand{\scriptC}{\mathcal{C}}
\newcommand{\scriptM}{\mathcal{M}}
\newcommand{\scriptV}{\mathcal{V}}
\newcommand{\scriptE}{\mathcal{E}}
\newcommand{\scriptB}{\mathcal{B}}

\newcommand{\scriptN}{\mathcal{N}}
\newcommand{\R}{\mathbb{R}}
\newcommand{\rightarrowdist}{\stackrel{d}{\longrightarrow}}

\renewcommand{\vec}[1]{\boldsymbol{\mathbf{#1}}}

\DeclareMathOperator{\Var}{Var}

\begin{document}

\title{Stochastic Block Transition Models for\\Dynamic Networks}
\renewcommand\Authfont{\small}
\renewcommand\Affilfont{\small}
\setlength{\affilsep}{3pt}
\author{Kevin S.~Xu}
\affil{Technicolor Research, 175 S. San Antonio Rd., Los Altos, CA 94022, USA 
\authorcr \url{kevinxu@outlook.com}}

\maketitle

\begin{abstract}
There has been great interest in recent years on 
statistical models for dynamic networks. 
In this paper, I propose a \emph{stochastic block transition model} (SBTM) 
for dynamic networks that is 
inspired by the well-known stochastic block model (SBM) for 
static networks and previous dynamic extensions of the SBM. 
Unlike most existing dynamic network models, it does \emph{not} make a 
hidden Markov 
assumption on the edge-level dynamics, allowing the presence or absence of 
edges to directly influence future edge probabilities while retaining the 
interpretability of the SBM. 
I derive an approximate inference procedure for the SBTM and 
demonstrate that it is significantly better at reproducing durations of 
edges in real social network data.
\end{abstract}

\section{Introduction}
Analysis of data in the form of networks has been a topic of 
interest across many disciplines, aided by the development of statistical 
models for networks. 
Many models have been proposed for static networks, where the data consist of 
a single observation of the network \citep{Goldenberg2009}. 
On the other hand, modeling dynamic networks is still in its infancy; much 
research on dynamic network modeling has appeared only in the past several 
years. 
Statistical models for static networks typically utilize a latent variable 
representation for the network; such models have been extended to dynamic 
networks by allowing the latent variables, which I refer to as \emph{states}, 
to evolve over time. 

This paper targets networks evolving in discrete time in which both nodes 
and edges can \emph{appear} and \emph{disappear} over time, such as dynamic 
networks of social interactions. 
Most existing dynamic network models assume a hidden Markov 
structure, where a snapshot of the network at any particular time is 
\emph{conditionally independent} from all previous snapshots given the current 
network states. 
Such an approach greatly simplifies the model and allows for tractable 
inference, but it may not be flexible enough to replicate 
certain observations from real network data, such as time durations of edges, 
which are often inaccurately reproduced by models with hidden Markov 
dynamics. 

In this paper I propose a \emph{stochastic block transition model} (SBTM) 
for dynamic networks, 
inspired by the well-known stochastic block model (SBM) for static networks. 
The approach generalizes two recent dynamic extensions of SBMs that utilize 
the hidden Markov assumption \citep{Yang2011,Xu2014a}. 
In the SBTM, the presence (or absence) of an edge between two nodes at any 
given time step \emph{directly influences} the probability that such an edge
would appear at the next time step. 

I demonstrate that, under the SBTM, the sample mean of a scaled version 
of the observed adjacency matrix at each time is 
asymptotically Gaussian. 
Taking advantage of this property, I develop an approximate 
inference procedure using a combination of an extended Kalman filter and 
a local search algorithm. 
I investigate the accuracy of the inference procedure via a simulation 
experiment. 
Finally I fit the SBTM to a real dynamic network of social interactions and 
demonstrate its ability to more accurately replicate 
edge durations while retaining the interpretability of the SBM.

\section{Related Work}
\label{sec:Related}
There has been significant research dedicated to statistical modeling of 
dynamic networks, mostly in the past several years. 
Much of the earlier work is covered in the excellent survey by 
\citet{Goldenberg2009}. 
Key contributions in this area include dynamic extensions of static network 
models including exponential random graph models \citep{Guo2007},
stochastic block models \citep{Xing2010,Ho2011,Ishiguro2010, 
Yang2011,Xu2014a},
continuous latent space models \citep{Sarkar2005,Sarkar2007, 
Hoff2011,Lee2011,Durante2013a},
and latent feature models \citep{Foulds2011,Heaukulani2013,Kim2013}. 

Several dynamic extensions of stochastic block models are related to 
this paper. 
\citet{Xing2010} and \citet{Ho2011} proposed dynamic extensions of a 
mixed-membership version of the SBM. 
\citet{Ishiguro2010} proposed a dynamic extension of the infinite relation 
model, which is a nonparametric version of the SBM. 
\citet{Yang2011} and \citet{Xu2014a} proposed dynamic extensions of the 
standard SBM; these models are closely related to the model proposed in this 
paper and are further discussed in Section \ref{sec:DSBM}. 

Most dynamic network models assume a hidden Markov structure. 
Specifically the network states follow Markovian 
dynamics, and it is assumed that a network snapshot is  
\emph{conditionally independent} of all past snapshots given the current 
states. 
While tractable, such an assumption may not be realistic in many settings, 
including dynamic networks of social interactions. 
For example, if two people interact with each other at some time, it 
may influence them to interact again in the near future. 
\citet{Viswanath2009} reported that over $80\%$ of pairs of Facebook users 
continued to interact one month after an 
initial interaction, and over $60\%$ continued after three months, suggesting 
that such an influence may be present. 

In hidden Markov dynamic network models, observing an edge influences the 
estimated probability 
of that edge re-occurring in the future only by affecting the estimated states 
corresponding to the edge, so the influence is weak. 
A \emph{stronger influence} can be incorporated by allowing the presence of a 
future edge to depend both on the current network states and 
on whether or not an edge is currently present. 
The model I propose satisfies this property. 
To the best of my knowledge, the only other dynamic network model satisfying 
this property is the latent feature propagation model proposed by 
\citet{Heaukulani2013}. 

\section{Stochastic Block Models}

\subsection{Static Stochastic Block Models}
A static network is represented by a graph over a set of nodes $\scriptV$ 
and a set of edges $\scriptE$. 
The nodes and edges are represented by a square adjacency matrix $W$, where 
an entry $w_{ij}=1$ denotes that an edge is present from node $i \in \scriptV$ 
to node $j \in \scriptV \setminus \{i\}$, and $w_{ij}=0$ denotes that no 
such edge is present. 
Unless otherwise specified, I assume directed graphs, i.e.~$w_{ij} \neq 
w_{ji}$ in general, with no self-edges, i.e.~$w_{ii}=0$. 
Let $\scriptC = \{\scriptC_1, \dots, \scriptC_k\}$ denote a partition of 
$\scriptV$ into $k$ classes. 
I use the notation $i \in a$ to denote that node $i$ belongs to class $a$. 
I represent the partition by a class membership vector $\vec{c}$, where 
$c_i = a$ is equivalent to $i \in a$.

A \emph{stochastic block model} (SBM) for a static network is defined as 
follows (adapted from Definition 3 in \citet{Holland1983}):

\begin{definition}[Stochastic block model]
\label{def:SBM}
Let $W$ denote a \emph{random} adjacency matrix for a static network, and let 
$\vec{c}$ denote a class membership vector. 
$W$ is generated according to a stochastic block model with respect 
to the membership vector $\vec{c}$ if and only if,

\begin{enumerate}
\item For any nodes $i \neq j$, the random variables $w_{ij}$ are 
statistically independent.

\item For any nodes $i \neq j$ and $i' \neq j'$, if 
$i$ and $i'$ are in the same class, i.e.~$c_i = c_{i'}$,
and $j$ and $j'$ are in the same class, i.e.~$c_j = c_{j'}$,
then the random variables 
$w_{ij}$ and $w_{i'j'}$ are identically distributed.
\end{enumerate}

\end{definition}

Let $\Theta \in [0,1]^{k \times k}$ denote the matrix of probabilities of 
forming edges between classes, which I refer to as the \emph{block 
probability matrix}.  
It follows from Definition \ref{def:SBM} and the requirement that $W$ be an 
adjacency matrix that $w_{ij} \sim 
\text{Bernoulli}(\theta_{ab})$, where $i \in a$ and $j \in b$. 

SBMs are used in both the \emph{a priori} setting, where class 
memberships are known or assumed, and the \emph{a posteriori} 
setting, where class memberships are estimated. 
Recent interest has focused on the more difficult a posteriori setting, 
which I assume in this paper.

\subsection{Dynamic Stochastic Block Models}
\label{sec:DSBM}
Consider a dynamic network evolving in discrete time steps where 
both nodes and edges could \emph{appear} or \emph{disappear} over time. 
Let $(\scriptV^t,\scriptE^t)$ denote a graph snapshot, where the superscript 
$t$ denotes the time step. 
Let $\scriptM^t$ denote a mapping from $\scriptV^t$, the set of 
nodes at time $t$, to the set of indices $\{1,\ldots,|\scriptV^t|\}$. 
Using the appropriate mapping $\scriptM^t$, one can represent a dynamic 
network using a sequence of adjacency matrices 
$W^{(T)} = \{W^1, \ldots, W^T\}$, and correspondence between rows and columns 
of different matrices can be established by inverting the mapping. 
In the remainder of this paper, I drop explicit reference to the mappings and 
assume that a node $i \in \scriptV^{t-1} \cap \scriptV^t$ is represented by 
row and column $i$ in both $W^{t-1}$ and $W^t$. 

I define a \emph{dynamic stochastic block model} for a time-evolving 
network in the following manner:

\begin{definition}[Dynamic stochastic block model]
\label{def:DSBM}
Let $W^{(T)}$ denote a \emph{random sequence} of $T$ adjacency matrices 
over the set 
of nodes $\scriptV^{(T)} = \cup_{t=1}^T \scriptV^t$, and 
let $\vec{c}^{(T)}$ denote a sequence of class membership vectors for these 
nodes. 
$W^{(T)}$ is generated according to a dynamic stochastic block model 
with respect to $\vec{c}^{(T)}$ if and only if
for each time $t$, $W^t$ is generated according to a static stochastic block 
model with respect to $\vec{c}^t$.

\end{definition}

This definition of a dynamic SBM encompasses 
dynamic extensions of SBMs previously proposed in the literature 
\citep{Yang2011,Xu2014a}, which model the sequence $W^{(T)}$ as observations 
from a hidden 
Markov-type model, where $W^t$ is conditionally independent of all past 
adjacency matrices $W^{(t-1)}$ given the parameters of the SBM at time $t$. 
I refer to these hidden Markov SBMs as HM-SBMs. 

\citet{Yang2011} proposed an HM-SBM that posits a Markov model on the class 
membership 
vectors $\vec{c}^t$ parameterized by a transition matrix that specifies the 
probability that any node in class $a$ at time $t$ switches to class $b$ at 
time $t+1$ for all $a,b,t$. 
The authors proposed an approximate inference procedure using a combination of 
Gibbs sampling and 
simulated annealing, which they refer to as probabilistic simulated 
annealing (PSA).

\citet{Xu2014a} proposed an HM-SBM that places 
a state-space model on the block probability 
matrices $\Theta^t$. 
The temporal evolution of these probabilities is governed by a linear dynamic 
system on the logits of the probabilities $\Psi^t = \log(\Theta^t 
/(1-\Theta^t))$, where the logarithms are applied entrywise. 
The authors performed approximate inference by using an extended Kalman filter 
augmented with a local search procedure, which was shown to perform 
competitively with the PSA procedure of \citet{Yang2011} in terms of accuracy 
but is about an order of magnitude faster. 

\section{Stochastic Block Transition Models}
\label{sec:SBTM}

One of the main disadvantages of using  
a hidden Markov-type approach for dynamic SBMs 
relates to the assumption that edges at time $t$ are conditionally independent 
from edges at previous times given the SBM parameters (states) at time $t$. 
Hence the probability distribution of edge durations is given by
\begin{equation*}
	\Pr(\text{duration} = d) \!=\! \left(1-\theta_{ab}^{t-1}\right) \!
		\theta_{ab}^t \cdots \theta_{ab}^{t+d-1} \!
		\left(1-\theta_{ab}^{t+d}\right)\!,
\end{equation*}
for an edge that first appeared at time $t$ and disappeared at $t+d$ 
where the nodes belong to classes $a$ and $b$ from times $t-1$ to $t+d$. 
Note that the edge durations are tied directly to the probabilities of 
forming edges at a given 
time $\theta_{ab}^t$, which control the densities of the blocks. 
Specifically, the presence or absence of an edge between two nodes at any 
particular time \emph{does not directly influence} the presence or absence of 
such an 
edge at a future time, which is undesirable in certain settings, as noted 
in Section \ref{sec:Related}. 

\subsection{Model Definition}
\label{sec:ModelDef}
I propose a dynamic network model where the edge durations are 
decoupled from the block densities, which 
allows for edges with long durations even in blocks with 
low densities. 
The main idea is as follows: for any pair of nodes $i \in a$ and $j \in b$ 
at both times $t-1$ and $t$ such that 
$w_{ij}^{t-1} = 1$, i.e.~there is an edge from $i$ to $j$ at time $t-1$, 
$w_{ij}^t$ are independent and identically distributed 
(iid). 
The same is true for $w_{ij}^{t-1} = 0$. 
Thus all edges in a block at 
time $t-1$ are equally likely to re-appear at time $t$, and 
non-edges in a block at 
time $t-1$ are equally likely to appear at time $t$. 
Since these sub-blocks are on the \emph{transitions} between time steps, I 
call this the \emph{stochastic block transition model} (SBTM). 

Let $i$ and $j$ denote nodes in classes $a$ and $b$, respectively, at both 
times $t-1$ and $t$, and define
\begin{gather}
\label{eq:TransProbNew}
\pi_{ab}^{t|0} = \Pr(w_{ij}^t=1|w_{ij}^{t-1}=0)\; \\
\label{eq:TransProbExist}
\pi_{ab}^{t|1} = \Pr(w_{ij}^t=1|w_{ij}^{t-1}=1).
\end{gather} 
Unlike in the hidden Markov SBM, where edges are formed iid with 
probabilities according 
to the block probability matrix $\Theta^t$, in the SBTM, edges are formed 
according to two block transition matrices: 
$\Pi^{t|0} = \big[\pi_{ab}^{t|0}\big]$, 
denoting the probability of \emph{forming new edges} within blocks, and 
$\Pi^{t|1} = \big[\pi_{ab}^{t|1}\big]$, 
denoting the probability of \emph{existing edges re-occurring} within blocks. 

The SBTM can accommodate nodes changing classes over time as well as 
new nodes entering the network. 
If a node was not present at time $t-1$, take its class membership at time 
$t-1$ to be $0$. 
I formally define the SBTM as follows:
\begin{definition}[Stochastic block transition model]
\label{def:SBTM}

Let $W^{(T)}$ and $\vec{c}^{(T)}$ denote the same quantities as in 
Definition \ref{def:DSBM}. 
$W^{(T)}$ is generated according to a stochastic block 
transition model with respect to $\vec{c}^{(T)}$ if and only if,

\begin{enumerate}
\item The initial adjacency matrix $W^1$ is generated according to a static 
SBM with respect to $\vec{c}^{1}$. 

\item At any given time $t$, for any nodes $i \neq j$, the random variables 
$w_{ij}^t$ are statistically independent.

\item At time $t \geq 2$, for any nodes $i \neq j$ such that 
$c_i^t = a$ and $c_j^t = b$ and for $u \in \{0,1\}$, 
\label{item:WprevEqual}
\begin{equation}
\label{eq:TransProb}
\Pr(w_{ij}^t = 1 | w_{ij}^{t-1} = u) = \xi_{ij}^t \pi_{ab}^{t|u}.
\end{equation}
\end{enumerate}

\end{definition}
The matrix of \emph{scaling factors} $\Xi^t = [\xi_{ij}^t]$ is used to scale 
the transition probabilities $\pi_{ab}^{t|0}$ and $\pi_{ab}^{t|1}$ to account 
for new nodes entering the network as well as existing nodes changing classes 
over time. 

I propose to choose the scaling factors $\xi_{ij}^t$ to satisfy the following 
properties:
\begin{enumerate}
\item If nodes $i \in a$ and $j \in b$ at both 
times $t-1$ and $t$, then $\xi_{ij}^t = 1$.
\label{item:ScaleSameClass}

\item The scaled transition probability is a valid probability, 
i.e.~$0 \leq \xi_{ij}^t \pi_{ab}^{t|u} \leq 1$ for all $i \neq j$ such that 
$c_i^t = a$, $c_j^t = b$, and $u \in \{0,1\}$.
\label{item:ScaleValid}

\item The marginal distribution of the adjacency matrix $W^t$ should follow a 
static SBM.
\label{item:ScaleSbm}
\end{enumerate}
Property \ref{item:ScaleSameClass} follows from the definition of the 
transition probabilities \eqref{eq:TransProbNew} and 
\eqref{eq:TransProbExist}. 
Property \ref{item:ScaleValid} ensures that the SBTM is a valid model. 
Finally, property \ref{item:ScaleSbm} provides the connection to the 
static SBM. 

\subsection{Derivation of Scaling Factors}
I derive an expression for the scaling factors that satisfies each of 
the three properties. 
Consider two nodes $i \in a'$ and $j \in b'$ 
at time $t-1$ and $i \in a$ and $j \in b$ at time $t$. 
Begin with the case where $a' = 0$ or $b' = 0$, indicating that either 
node $i$ or $j$, respectively, was not present at time $t-1$.
For this case, $w_{ij}^{t-1}=0$ so
\begin{equation*}
\Pr(w_{ij}^t = 1) = \Pr(w_{ij}^t = 1 | w_{ij}^{t-1} = 0) 
	= \xi_{ij}^t \pi_{ab}^{t|0}
\end{equation*}
Property \ref{item:ScaleSameClass} does not apply. 
In order for property \ref{item:ScaleSbm} to hold, $\Pr(w_{ij}^t=1)$ must 
be equal to $\theta_{ab}^t$. 
Thus $\xi_{ij}^t = \theta_{ab}^t / \pi_{ab}^{t|0}$. 
Note that this also satisfies property \ref{item:ScaleValid} because 
$\theta_{ab}^t$ is a valid probability. 

Next consider the case where $a',b' \neq 0$, i.e.~both nodes were present at 
the previous time. 
Then
\begin{align}
&\Pr(w_{ij}^t = 1) \nonumber\\
=\; &\Pr(w_{ij}^t = 1 | w_{ij}^{t-1} = 0) \Pr(w_{ij}^{t-1} = 0) 
	\nonumber\\
\label{eq:LawTotalProb}
&\qquad\qquad\qquad + \Pr(w_{ij}^t = 1 | w_{ij}^{t-1} = 1) \Pr(w_{ij}^{t-1} = 1) \\
\label{eq:MargProbGen}
=\; &\xi_{ij}^{t|0} \pi_{ab}^{t|0} (1-\theta_{a'b'}^{t-1}) + 
	\xi_{ij}^{t|1} \pi_{ab}^{t|1} \theta_{a'b'}^{t-1},
\end{align}
where \eqref{eq:MargProbGen} follows from substituting \eqref{eq:TransProb} 
into \eqref{eq:LawTotalProb} and by letting the scaling factor
\begin{equation}
\xi_{ij}^t = 
\begin{cases}
\xi_{ij}^{t|0}, & \text{if } w_{ij}^{t-1} = 0 \\
\xi_{ij}^{t|1}, & \text{if } w_{ij}^{t-1} = 1
\end{cases}.
\end{equation} 
According to property \ref{item:ScaleSbm}, $\Pr(w_{ij}^t = 1) 
= \theta_{ab}^t$. 
Hence one must choose the scaling factor $\xi_{ij}^t$ such that this is 
the case. 
If $a = a'$ and $b = b'$, i.e. neither node changed class between time steps, 
then $\xi_{ij}^t = 1$ from property \ref{item:ScaleSameClass}, so 
\eqref{eq:MargProbGen} becomes 
\begin{equation}
\label{eq:ThetaRecurs}
\theta_{ab}^t = \pi_{ab}^{t|0} (1-\theta_{ab}^{t-1}) + \pi_{ab}^{t|1} 
	\theta_{ab}^{t-1}.
\end{equation}

For the general case where $a \neq a'$ or $b \neq b'$, 
I first identify a range of choices for the scaling factor 
$\xi_{ij}^t$ that satisfy properties \ref{item:ScaleValid} and 
\ref{item:ScaleSbm}, then I select a particular choice that satisfies 
property \ref{item:ScaleSameClass}. 
Property \ref{item:ScaleValid} implies the following inequalities:
\begin{gather}
\label{eq:ScaleNewIneq}
0 \leq \xi_{ij}^{t|0} \leq 1/\pi_{ab}^{t|0} \; \\
\label{eq:ScaleExistIneq}
0 \leq \xi_{ij}^{t|1} \leq 1/\pi_{ab}^{t|1}.
\end{gather}
Meanwhile property \ref{item:ScaleSbm} implies that
\begin{equation}
\label{eq:ThetaRecursGen}
\theta_{ab}^t = \xi_{ij}^{t|0} \pi_{ab}^{t|0} (1-\theta_{a'b'}^{t-1}) + 
	\xi_{ij}^{t|1} \pi_{ab}^{t|1} \theta_{a'b'}^{t-1}.
\end{equation}
Re-arrange \eqref{eq:ThetaRecursGen} to isolate $\xi_{ij}^{t|1}$ and 
substitute into \eqref{eq:ScaleExistIneq} to obtain 
\begin{equation}
\label{eq:ScaleNewIneq2}
\frac{\theta_{ab}^t - \theta_{a'b'}^{t-1}} {\pi_{ab}^{t|0} 
	(1-\theta_{a'b'}^{t-1})} \leq \xi_{ij}^{t|0} \leq \frac{\theta_{ab}^t} 
	{\pi_{ab}^{t|0} (1-\theta_{a'b'}^{t-1})}.
\end{equation}
Combine \eqref{eq:ScaleNewIneq}, \eqref{eq:ThetaRecursGen}, and 
\eqref{eq:ScaleNewIneq2} to arrive at necessary conditions on 
$\pi_{ab}^{t|0}$ in order to satisfy properties \ref{item:ScaleValid} and 
\ref{item:ScaleSbm}:
\begin{equation}
\label{eq:ScaleNewNecess}
\alpha(a',b') \leq \xi_{ij}^{t|0} \leq \beta(a',b'),
\end{equation}
where the upper and lower bounds are functions of $a'$ and $b'$, the 
classes for $i$ and $j$, respectively, at time $t-1$ and are given by
\begin{gather}
\label{eq:LBoundGen}
\alpha(a',b') = \max\left(0,\frac{\theta_{ab}^t - \theta_{a'b'}^{t-1}} 
	{\pi_{ab}^{t|0} (1-\theta_{a'b'}^{t-1})}\right) \\
\label{eq:UBoundGen}
\beta(a',b') = \min\left(\frac{1}{\pi_{ab}^{t|0}},\frac{\theta_{ab}^t} 
	{\pi_{ab}^{t|0} (1-\theta_{a'b'}^{t-1})}\right)
\end{gather}
From \eqref{eq:ScaleNewNecess}--\eqref{eq:UBoundGen}, it follows that 
\begin{equation}
\label{eq:ScaleNewSolRange}
\xi_{ij}^{t|0} = \alpha(a',b') + \frac{\beta(a',b') - \alpha(a',b')}
	{\gamma(a',b')}
\end{equation}
is a valid solution for any $\gamma(a',b') \geq 1$.

In order to satisfy property \ref{item:ScaleSameClass} as well, 
$\xi_{ij}^{t|0}$ must equal $1$ if $a' = a$ and $b' = b$, i.e.~neither 
node changed class between time steps. 
This is accomplished by choosing 
\begin{equation}
\label{eq:ScaleNewDenom}
\gamma(a',b') = \frac{\beta(a,b) - \alpha(a,b)} {1 - \alpha(a,b)}.
\end{equation}
Notice that the arguments in $\alpha(\cdot)$ and $\beta(\cdot)$ are the 
current classes $a$ and $b$, regardless of the previous classes. 

The assignment for $\xi_{ij}^{t|0}$ is thus obtained by 
substituting \eqref{eq:ScaleNewDenom} into \eqref{eq:ScaleNewSolRange}. 
This value can then be substituted into \eqref{eq:ThetaRecursGen} to obtain 
the assignment for $\xi_{ij}^{t|1}$. 

\begin{prop}
\label{prop:ScaleFac}
The scaling factor assignment given by \eqref{eq:ThetaRecursGen}, 
\eqref{eq:ScaleNewSolRange}, and \eqref{eq:ScaleNewDenom} satisfies the 
three properties specified in Section \ref{sec:ModelDef}.
\end{prop}

\begin{proof}
Begin with property \ref{item:ScaleSameClass}. 
Let $i \in a$ and $j \in b$ at both times $t-1$ and $t$. 
From \eqref{eq:ScaleNewSolRange} and \eqref{eq:ScaleNewDenom},
\begin{align}
\xi_{ij}^{t|0} &= \alpha(a,b) + \frac{\beta(a,b) - \alpha(a,b)}{\gamma(a,b)} 
	\nonumber\\
&= \alpha(a,b) + (1 - \alpha(a,b)) \nonumber\\
\label{eq:ScaleNew1}
&= 1.
\end{align}
Substituting \eqref{eq:ScaleNew1} and \eqref{eq:ThetaRecurs} into 
\eqref{eq:ThetaRecursGen},
\begin{equation*}
\xi_{ij}^{t|1} = \frac{\pi_{ab}^{t|0}(\theta_{ab}^{t-1}-1)} {\pi_{ab}^{t|1} 
	\theta_{ab}^{t-1}} + \frac{\pi_{ab}^{t|0} 
	(1-\theta_{ab}^{t-1}) + \pi_{ab}^{t|1} \theta_{ab}^{t-1}} {\pi_{ab}^{t|1} 
	\theta_{ab}^{t-1}} = 1.
\end{equation*}
Thus property \ref{item:ScaleSameClass} is satisfied.

From the derivation of the scaling factor assignment, it was shown that 
properties \ref{item:ScaleValid} and \ref{item:ScaleSbm} are satisfied 
provided $\gamma(a',b') \geq 1$ for all $(a',b')$. 
From \eqref{eq:ScaleNewDenom}, this is true if and only if $\beta(a,b) 
\geq 1$ for all $(a,b)$. 
From \eqref{eq:UBoundGen}, $\beta(a,b) \geq 1/\pi_{ab}^{t|0} \geq 1$
because $\pi_{ab}^{t|0}$ is a probability and hence must be between $0$ and 
$1$, and
\begin{equation*}
\beta(a,b) \geq \frac{\theta_{ab}^t}{\pi_{ab}^{t|0} (1-\theta_{ab}^{t-1})} 
= 1 + \frac{\pi_{ab}^{t|1} \theta_{ab}^{t-1}}{\pi_{ab}^{t|0} 
(1-\theta_{ab}^{t-1})} \geq 1,
\end{equation*}
where the equality follows from \eqref{eq:ThetaRecurs}, and the final 
inequality results from $\pi_{ab}^{t|0}$, $\pi_{ab}^{t|1}$, and 
$\theta_{ab}^{t-1}$ all being probabilities and hence between $0$ and $1$. 
Thus properties \ref{item:ScaleValid} and \ref{item:ScaleSbm} are also 
satisfied.
\end{proof}

\begin{prop}
\label{prop:SbtmDynSbm}
An SBTM with respect to $\vec{c}^{(T)}$ 
satisfying such an assumption is a dynamic SBM; that is, any 
sequence $W^{(T)}$ generated by the SBTM also satisfies 
the requirements of a dynamic SBM.
\end{prop}

Proposition \ref{prop:SbtmDynSbm} holds trivially from property 
\ref{item:ScaleSbm}, which is satisfied due to Proposition 
\ref{prop:ScaleFac}. 
Both the SBTM and HM-SBM are dynamic SBMs; the main difference between the two 
is that, 
under the SBTM, the presence or absence of an edge between two nodes at a 
particular time \emph{does} affect the presence or absence of such an edge at 
a future time as indicated by \eqref{eq:TransProb}.

\subsection{State Dynamics}
The SBTM, as defined in Definition \ref{def:SBTM}, does not specify the 
model governing the dynamics of the sequence of adjacency matrices $W^{(T)}$ 
aside from the dependence of $W^t$ on $W^{t-1}$ specified in requirement 
\ref{item:WprevEqual}. 
To complete the model, I use a linear dynamic system on the logits of the 
probabilities, similar to \citet{Xu2014a}. 
Unlike \citet{Xu2014a}, however, the states of the system would be the logits 
of the block transition matrices $\Pi^{t|0}$ and $\Pi^{t|1}$. 

Let $\vec{x}$ denote the vectorized equivalent of a matrix $X$, obtained by 
stacking columns on top of one another, so that $\vec{\pi}^{t|0}$ and 
$\vec{\pi}^{t|1}$ are the vectorized equivalents of $\Pi^{t|0}$ and 
$\Pi^{t|1}$, respectively. 
The states of the system can then be expressed as a vector 
\begin{equation}
	\label{eq:StateDef}
	\vec{\psi}^t = 
	\begin{bmatrix}
	\log(\vec{\pi}^{t|0}/(1-\vec{\pi}^{t|0})) \\
	\log(\vec{\pi}^{t|1}/(1-\vec{\pi}^{t|1}))
	\end{bmatrix},
\end{equation}
resulting in the dynamic linear system
\begin{equation}
	\label{eq:LinearSys}
	\vec{\psi}^t = F^t \vec{\psi}^{t-1} + \vec{v}^t,
\end{equation}
where $F^t$ is the state transition model applied to the previous state, and 
$\vec{v}^t$ is a random vector of zero-mean Gaussian entries, 
commonly referred to as process noise, with covariance matrix $\Gamma^t$. 
Note that \eqref{eq:LinearSys} is the same dynamic system equation as in 
\citet{Xu2014a}, only with a different definition 
\eqref{eq:StateDef} for the state vector. 

\section{Model Inference}
\label{sec:Inference}

\subsection{Asymptotic Distribution of Observations}
The inference procedure for the dynamic SBM of \citet{Xu2014a} 
utilized a Central Limit Theorem (CLT) approximation for the block densities, 
which 
are scaled sums of independent, identically distributed Bernoulli random 
variables $w_{ij}^t$. 
Such an approach cannot be used for the SBTM because blocks no longer consist 
of identically distributed variables $w_{ij}^t$ due to the dependency 
between $W^t$ and $W^{t-1}$. 
Furthermore, the presence of the scaling factors $\xi_{ij}^t$ 
in the transition probabilities \eqref{eq:TransProb} ensure that 
$w_{ij}^t$ are not identically distributed even after conditioning on 
$w_{ij}^{t-1}$. 

I show, however, that the sample mean of a scaled version of the 
adjacencies, is asymptotically Gaussian. 
For $a,b \in \{1,\ldots,k\}$ and $u \in \{0,1\}$, let
\begin{equation*}
\scriptB_{ab}^{t|u} = \{(i,j): i \neq j, c_i^t = a, c_j^t = b, w_{ij}^{t-1} 
	= u\}.
\end{equation*}
Note that 
$\scriptB_{ab}^{t|0}$ denotes the set of non-edges in block $(a,b)$ 
at time $t-1$, which is 
also the set of \emph{possible new edges} at time $t$, and 
$\scriptB_{ab}^{t|1}$ denotes the set of edges in block $(a,b)$ at time 
$t-1$, which is also the set of \emph{possible re-occurring edges} at time 
$t$. 
Let 
\begin{equation*}
m_{ab}^{t|u} = \sum_{(i,j) \in \scriptB_{ab}^{t|u}} \frac{w_{ij}^t} 
	{\xi_{ij}^t}
\end{equation*}
and $n_{ab}^{t|u} = \big|\scriptB_{ab}^{t|u}\big|$.
$m_{ab}^{t|0}$ and $m_{ab}^{t|1}$ denote the scaled number of new and  
re-occurring edges, respectively, within block $(a,b)$ at time $t$, while 
$n_{ab}^{t|0}$ and $n_{ab}^{t|1}$ denote the number of \emph{possible} new 
and re-occurring edges, respectively. 
The following theorem shows that the sample mean of the scaled adjacencies 
within $\scriptB_{ab}^{t|u}$ is asymptotically Gaussian as the block 
size increases. 

\begin{theorem}
\label{thm:AsyGauss}
The sample mean of the scaled adjacencies
\begin{equation*}
\frac{m_{ab}^{t|u}}{n_{ab}^{t|u}} = \frac{1}{n_{ab}^{t|u}} \sum_{(i,j) \in 
	\scriptB_{ab}^{t|u}} \frac{w_{ij}^t}{\xi_{ij}^t} 
	\rightarrow\scriptN\left(\pi_{ab}^{t|u}, \left(\frac{s_{ab}^{t|u}} 
	{n_{ab}^{t|u}} \right)^2\right)
\end{equation*}
in distribution as $n_{ab}^{t|u} \rightarrow \infty$, where
\begin{equation}
s_{ab}^{t|u}
= \left[\pi_{ab}^{t|u} \sum_{(i,j) \in \scriptB_{ab}^{t|u}} \frac{1}{\xi_{ij}^t} 
	- n_{ab}^{t|u} \left(\pi_{ab}^{t|u}\right)^2\right]^{1/2}.
\label{eq:CondVar}
\end{equation}
\end{theorem}

\begin{proof}
The scaled adjacencies $w_{ij}^t/\xi_{ij}^t$ are independent, but not 
identically distributed, so the classical CLT no longer applies. 
However, the Lyapunov CLT can be applied provided Lyapunov's condition 
is satisfied \citep{Billingsley1995}. 
Let $\Var_u(\cdot)$ denote the conditional variance $\Var(\cdot 
| w_{ij}^{t-1}=u)$. 
The conditional variance of the scaled adjacencies is given by 
\begin{equation*}
\Var_u\left(\frac{w_{ij}^t}{\xi_{ij}^t}\right) 
= \left(\frac{1}{\xi_{ij}^t} \right)^2 \left(\xi_{ij}^t \pi_{ab}^{t|u}\right) 
	\left(1-\xi_{ij}^t \pi_{ab}^{t|u}\right) 
= \frac{\pi_{ab}^{t|u}} {\xi_{ij}^t} - \big(\pi_{ab}^{t|u}\big)^2.
\end{equation*}
Thus
\begin{equation*}
\sum_{(i,j) \in \scriptB_{ab}^{t|u}} 
	\!\!\Var_u\left(\frac{w_{ij}^t}{\xi_{ij}^t}\right) 
\,=\, \pi_{ab}^{t|u} \sum_{(i,j) \in \scriptB_{ab}^{t|u}} \frac{1}{\xi_{ij}^t} 
	- n_{ab}^{t|u} \left(\pi_{ab}^{t|u}\right)^2 = \big(s_{ab}^{t|u}\big)^2,
\end{equation*}
where $s_{ab}^{t|u}$ was defined in \eqref{eq:CondVar}. 
In this setting, Lyapunov's condition specifies that for some $\delta > 0$,
\begin{equation*}
\lim_{n_{ab}^{t|u} \rightarrow \infty} \frac{1}{\big(s_{ab}^{t|u}\big) 
	^{2+\delta}} \sum_{(i,j) \in \scriptB_{ab}^{t|u}} \!\! E_u\left[\left| 
	\frac{w_{ij}^t}{\xi_{ij}^t} - \pi_{ab}^{t|u}\right|^{2+\delta} \right] 
	= 0,
\end{equation*}
where $E_u[\cdot]$ denotes the conditional expectation
$E[\cdot | w_{ij}^{t-1}=u]$. 

I demonstrate that Lyapunov's condition is satisfied for $\delta=2$. 
First note that, although there are an infinite number of terms in the 
summation (in the limit), there are a finite number of unique terms. 
Specifically $w_{ij}^t \in \{0,1\}$, and $\xi_{ij}^t$ depends only on $i,j$ 
through their current and previous class memberships $a$, $b$, $a'$, and $b'$, 
which are all in $\{0, 1, \ldots, k\}$. 
Hence
\begin{align}
\label{eq:CondMean}
&\,\frac{1}{\big(s_{ab}^{t|u}\big)^4} \sum_{(i,j) \in \scriptB_{ab}^{t|u}} 
	\!\! E_u\left[\left(\frac{w_{ij}^t}{\xi_{ij}^t} - \pi_{ab}^{t|u} 
	\right)^4 \right] \\
\leq &\,\frac{n_{ab}^{t|u}}{\big(s_{ab}^{t|u}\big)^4} \max_{(i,j) \in 
	\scriptB_{ab}^{t|u}} E_u\left[\left(\frac{w_{ij}^t}{\xi_{ij}^t} - 
	\pi_{ab}^{t|u} \right)^4 \right] \nonumber\\
= &\,\frac{1}{O\big(n_{ab}^{t|u}\big)} \nonumber,
\end{align}
where the last equality follows from \eqref{eq:CondVar}. 
Thus \eqref{eq:CondMean} approaches $0$ as $n_{ab}^{t|u} \rightarrow \infty$, 
and Lyapunov's condition is satisfied. 
The Lyapunov CLT states that
\begin{equation*}
\frac{1}{s_{ab}^{t|u}} \sum_{(i,j) \in \scriptB_{ab}^{t|u}} \left( 
	\frac{w_{ij}^t}{\xi_{ij}^t} - \pi_{ab}^{t|u}\right) \rightarrowdist 
	\scriptN(0,1) \\
\end{equation*}
where $\rightarrowdist$ denotes convergence in distribution.
By rearranging terms one obtains the desired result.
\end{proof}

\subsection{State-space Model Formulation}
Theorem \ref{thm:AsyGauss} shows that the sample means $m_{ab}^{t|u} 
/n_{ab}^{t|u}$ are asymptotically Gaussian. 
Assume  they are indeed Gaussian. 
Stack these entries to form the observation vector
\begin{align}
	\vec{y}^t &= 
		\begin{bmatrix}
		\displaystyle \frac{m_{11}^{t|0}} {n_{11}^{t|0}} \;
		\cdots \;
		\frac{m_{kk}^{t|0}} {n_{kk}^{t|0}} \;
		\frac{m_{11}^{t|1}} {n_{11}^{t|1}} \;
		\cdots \;
		\frac{m_{kk}^{t|1}} {n_{kk}^{t|1}}
		\end{bmatrix}^T \nonumber\\
	\label{eq:Obs_model_SBM}
	&= h\left(\vec{\psi}^t\right) + \vec{z}^t,
\end{align}
where the function $h: \R^{2k^2} \rightarrow \R^{2k^2}$ is defined by 
\begin{equation}
	\label{eq:Logistic_fn}
	h_i(\vec{x}) = 1/(1+e^{-x_i}), 
\end{equation}
i.e.~the logistic sigmoid applied to each entry of $\vec{x}$, 
$\vec{\psi}^t$ was defined in \eqref{eq:StateDef}, and 
$\vec{z}^t \sim \scriptN (\vec{0},\Sigma^t)$, where $\Sigma^t$ is a 
diagonal matrix with entries given by $\big(s_{ab}^{t|u} / n_{ab}^{t|u}
\big)^2$. 

Equations \eqref{eq:LinearSys} and \eqref{eq:Obs_model_SBM} 
form a non-linear (due to the logistic function $h(\cdot)$) dynamic system 
with zero-mean Gaussian 
observation and process noise terms $\vec{z}^t$ and $\vec{v}^t$, respectively. 
Assume that the initial state is also Gaussian, 
i.e.~$\vec{\psi}^1 \sim \scriptN\left(\vec{\mu}^1, \Gamma^1\right)$, and that 
$\{\vec{\psi}^1, \vec{v}^2, 
\ldots, \vec{v}^t, \vec{z}^2, \ldots, \vec{z}^t\}$ are mutually independent. 
If \eqref{eq:Obs_model_SBM} was linear, then the optimal estimate 
for $\vec\psi^t$ given 
observations $\vec{y}^{(t)}$ in terms 
of minimum mean-squared error and maximum a posteriori probability (MAP) 
would be given by the Kalman filter. 
Due to the non-linearity, I apply the extended Kalman filter (EKF), which 
linearizes the dynamics about the predicted state and results in a 
\emph{near-optimal} estimate (in the MAP sense) when the estimation errors are 
small enough to make the linearization accurate. 
The EKF was used for inference in systems of the form of 
\eqref{eq:LinearSys} and \eqref{eq:Obs_model_SBM} in \citet{Xu2014a}.

\subsection{Inference Procedure}
Once the vector of sample means $\vec{y}^t$ is obtained, a near-optimal 
estimate of the state vector $\vec{\psi}^t$ can be obtained using the EKF. 
In order to compute the sample means $\vec{y}^t$, however, one needs to first 
estimate the following quantities:
\begin{enumerate}
\item The unknown hyperparameters $(\vec{\mu}^1,\Gamma^1,\Sigma^t,\Gamma^t)$ 
of the state-space model \eqref{eq:LinearSys} and \eqref{eq:Obs_model_SBM}.
\label{item:SsmParam}

\item The vector of class memberships $\vec{c}^t$.
\label{item:ClassMem}

\item The matrix of scaling factors $\Xi^t$.
\label{item:ScaleFac}
\end{enumerate}

Methods for estimating items \ref{item:SsmParam} and \ref{item:ClassMem} are 
discussed in \citet{Xu2014a}. 
Item \ref{item:SsmParam} can be addressed using standard methods for 
state-space models, typically alternating 
between state and hyperparameter estimation \citep{Nelson2000}. 
Item \ref{item:ClassMem} is handled by alternating between a local search 
(hill climbing) algorithm to estimate class memberships and the EKF to 
estimate the edge transition probabilities $\Pi^{t|0}$ and $\Pi^{t|1}$.

The main difference between the inference procedures of the HM-SBM and the 
SBTM proposed in this paper involves item \ref{item:ScaleFac}. 
The matrix of scaling factors $\Xi^t$ is a function of the 
marginal edge probabilities at the current and previous times ($\Theta^t$ and 
$\Theta^{t-1}$, respectively) as well as the current probabilities of new and 
existing edges ($\Pi^{t|0}$ and $\Pi^{t|1}$, respectively). 
$\Theta^t$ can be computed from the other three quantities from 
\eqref{eq:ThetaRecurs}. 

I propose to use plug-in estimates of $\Theta^{t-1}$, $\Pi^{t|0}$, and 
$\Pi^{t|1}$ to estimate the scaling matrix $\Xi^t$. 
From property \ref{item:ScaleSameClass} in Section \ref{sec:ModelDef}, 
$\xi_{ij}^t = 1$ for all pairs of nodes that do not change classes between 
time steps. 
Thus it is only necessary to estimate the remaining entries of $\Xi^t$. 
Recall from \eqref{eq:StateDef} that the state vector $\vec{\psi}^t$ 
consists of logits of the probabilities of forming 
new edges $\vec{\pi}^{t|0}$ and the probabilities of existing edges 
re-occurring $\vec{\pi}^{t|1}$. 
Hence $\hat{\vec{\psi}}^{t|t-1}$, the EKF prediction of the state vector at 
time $t$ given observations up to time $t-1$ can be used to compute the 
plug-in estimates $\hat{\Pi}^{t|0}$ and $\hat{\Pi}^{t|1}$. 
The recursion is initialized at time $2$ using the maximum-likelihood 
(ML) estimate $\hat{\Theta}^{1}$ obtained from $W^1$. 
The spectral clustering procedure of \citet{Sussman2012} can be used to 
initialize the class assignments for the local search at time $1$. 
A sketch of the entire inference procedure is shown in Algorithm 
\ref{alg:SbtmInfer}. 

\begin{algorithm}[t]
\caption{SBTM inference procedure}
\label{alg:SbtmInfer}

At time step $1$:

\begin{algorithmic}[1]
\STATE Initialize estimated class assignment using spectral clustering on 
$W^1$

\STATE Compute ML estimates $\hat{\vec{c}}^1$ and $\hat{\Theta}^1$ by 
local search

\STATE Compute predicted state vector $\hat{\vec{\psi}^{2|1}}$ at time step 
$2$ using EKF predict phase
\end{algorithmic}

At time step $t>1$:

\begin{algorithmic}[1]
\STATE Initialize estimated class assignment $\hat{\vec{c}}^t \leftarrow
\hat{\vec{c}}^{t-1}$

\REPEAT[Local search (hill climbing) algorithm]
	\FORALL{neighboring class assignments}
		\STATE Compute plug-in estimate $\hat{\Xi}^t$ of scaling matrix using 
		$\hat{\Theta}^{t-1}$, EKF predicted state $\hat{\vec{\psi}}^{t|t-1}$, 
		and current class assignment
		
		\STATE Compute plug-in estimate $\hat{\vec{y}}^t$ of sample means 
		using $\hat{\Xi}^t$, $W^t$, and current class assignment
		
		\STATE Compute estimate $\hat{\vec{\psi}}^{t|t}$ of state vector using 
		EKF update phase
	\ENDFOR
\UNTIL{reached local maximum of posterior density}

\STATE Compute predicted state vector $\hat{\vec{\psi}^{t+1|t}}$ at time step 
$t+1$ using EKF predict phase

\end{algorithmic}
\end{algorithm}

\section{Experiments}
\label{sec:Expts}

\subsection{Simulated Networks}
In this experiment I generate synthetic networks in a manner similar to 
a simulation experiment in \citet{Yang2011} and \citet{Xu2014a}, except with 
the stochastic block transition model rather than the hidden Markov 
stochastic block model. 
The network consists of $128$ nodes initially split into $4$ classes of $32$ 
nodes each. 
The edge probabilities for blocks at the initial time step are chosen 
to be $\theta_{aa}^1 = 0.2580$ and $\theta_{ab}^1 = 0.0834$ for 
$a,b = 1, 2, 3, 4; a \neq b$. 
The mean $\vec{\mu}^1$ is chosen such that $\pi_{aa}^{1|0} = 0.1$, 
$\pi_{ab}^{1|0} = 0.05, a \neq b$, $\pi_{aa}^{1|1} = 0.7$, and 
$\pi_{ab}^{1|1} = 0.45, a \neq b$. 
The covariance $\Gamma^1$ for the initial state is 
chosen to be a scaled identity 
matrix $0.04I$. 
The state vector $\vec{\psi}^t$ evolves according to a Gaussian random walk 
model, i.e.~$F^t = I$ in \eqref{eq:LinearSys}. 
$\Gamma^t$ is constructed such that $\gamma_{ii}^t = 0.01$ and 
$\gamma_{ij}^t = 0.0025$ for $i \neq j$. 
$10$ time steps are generated, and at each time step, $10\%$ of the nodes are 
randomly selected to 
leave their class and are randomly assigned to one of the other three classes. 
For consistency with \citet{Yang2011} and \citet{Xu2014a}, I generate 
undirected graph snapshots in this experiment. 

\begin{figure}[tp]
\centering
\subfloat[True classes and scaling] 
	{\label{fig:QQ128TrueTrue}
	\includegraphics[width=2.3in]{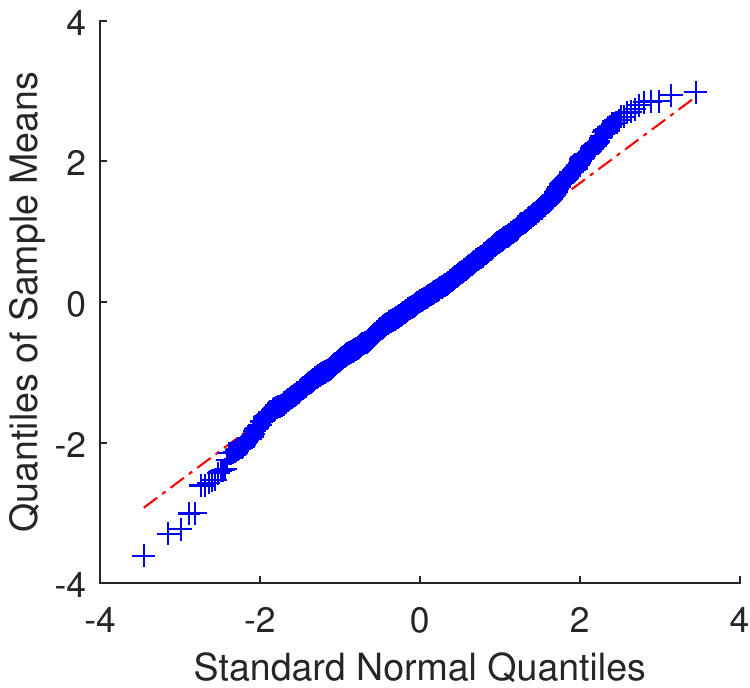}}
\quad
\subfloat[True classes, estimated scaling] 
	{\label{fig:QQ128TrueEst}
	\includegraphics[width=2.3in]{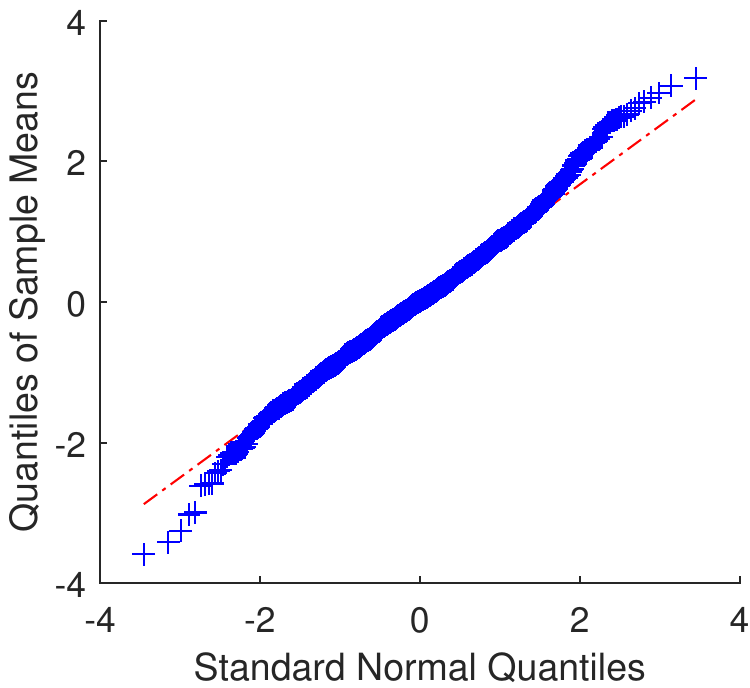}}
\quad
\subfloat[Estimated classes and scaling] 
	{\label{fig:QQ128EstEst}
	\includegraphics[width=2.3in]{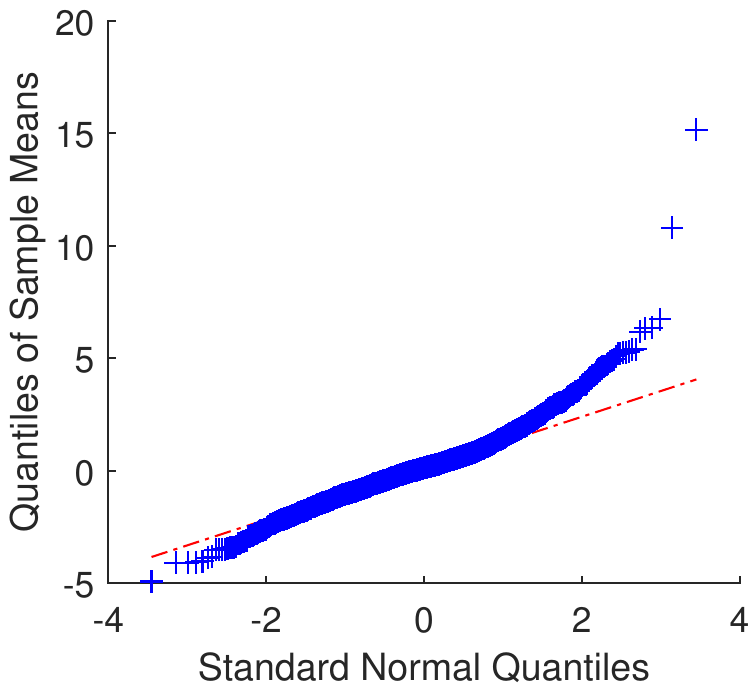}}
\caption[Q-Q plots of standardized sample means for $10$ runs of the simulated 
networks experiment under three levels of estimation] 
{Q-Q plots of standardized sample means $\vec{y}^t$ for $10$ runs of the 
simulated networks experiment under three levels of estimation. 
With \subref{fig:QQ128TrueTrue} true classes and scaling factors, 
$\vec{y}^t$ is close to the asymptotic Gaussian distribution predicted 
by Theorem \ref{thm:AsyGauss}. 
Even with \subref{fig:QQ128TrueEst} estimated scaling factors, 
$\vec{y}^t$ is still close to the asymptotic Gaussian distribution. 
When \subref{fig:QQ128EstEst} class memberships are also estimated, 
$\vec{y}^t$ is heavier tailed due to the errors in the estimated 
classes.}
\label{fig:QQ128}
\end{figure}

I begin by checking the validity of the asymptotic Gaussian distribution of 
the scaled sample means $\vec{y}^t$. 
In this simulation experiment, the population means and standard deviations 
for $\vec{y}^t$ are known and are used to standardize $\vec{y}^t$. 
Q-Q plots for the standardized $\vec{y}^t$ are shown in Figure 
\ref{fig:QQ128}. 
Figure \ref{fig:QQ128TrueTrue} shows the distribution of $\vec{y}^t$ when 
both the true classes and true scaling factors (calculated using the true 
states) are used. 
Notice that the empirical distribution is close to the asymptotic Gaussian 
distribution, with only slightly heavier tail. 
Experimentally I find that this deviation decreases as the block sizes 
increase, as one would expect from Theorem \ref{thm:AsyGauss}. 

Figure \ref{fig:QQ128TrueEst} shows that the distribution of $\vec{y}^t$ 
is roughly the same when using estimated scaling factors along with the 
true classes, which is an encouraging result and suggests that the EKF-based 
inference procedure would likely work well in the a priori block model 
setting. 
Figure \ref{fig:QQ128EstEst} shows that the distribution of $\vec{y}^t$ 
when using both estimated scaling factors and classes is significantly more 
heavy-tailed. 
Since this is not seen in Figure \ref{fig:QQ128TrueEst}, I conclude that it 
is due to errors in the class estimation, which causes the distribution 
of $\vec{y}^t$ to deviate from the asymptotically Gaussian distribution when 
using true classes. 
The heavier tails suggest that perhaps a more robust filter, such as a 
filter that assumes Student-t distributed observations, may provide more 
accurate estimates in the a posteriori setting.

\begin{figure}[t]
\centering
\includegraphics[width=3.5in]{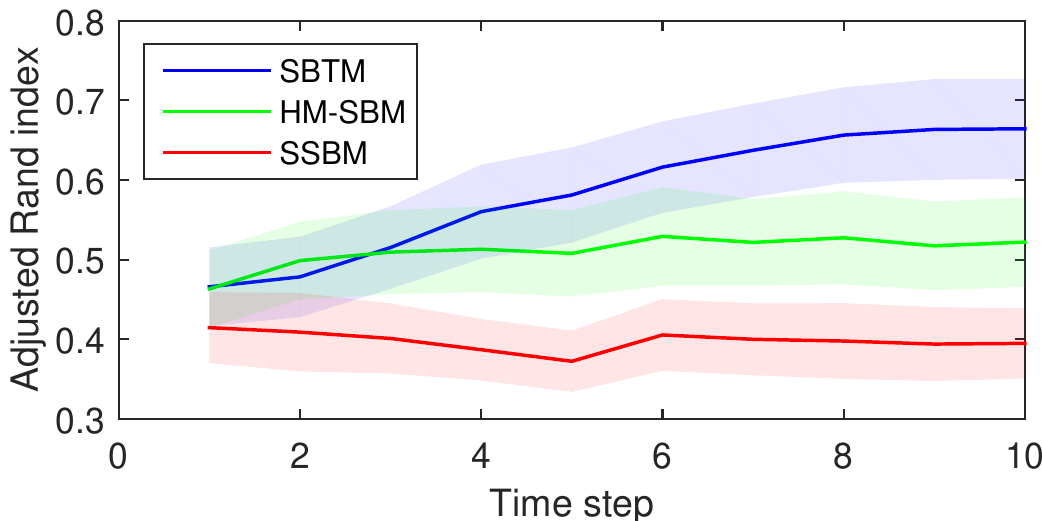}
\caption{Adjusted Rand indices with $95\%$ confidence bands for the stochastic 
block transition model (SBTM), hidden Markov stochastic block model (HM-SBM), 
and static stochastic block model (SSBM) on $50$ runs of the simulated 
networks experiment.}
\label{fig:SimRand}
\end{figure}

Figure \ref{fig:SimRand} shows a comparison of the class estimation 
accuracies, measured by the adjusted Rand indices \citep{Hubert1985}, 
of three different inference algorithms: the EKF-based algorithm 
for the SBTM proposed in this paper, the EKF algorithm for the HM-SBM 
\citep{Xu2014a}, and a static SBM fit using spectral clustering on each 
snapshot. 
As one might expect, the static SBM approach does not improve as more 
time snapshots are provided. 
The poorer performance of the HM-SBM approach compared to the proposed SBTM 
approach is also not too surprising since the dynamics on the marginal block 
probabilities no longer follow a dynamic linear system as assumed by 
\citet{Xu2014a}. 
The SBTM approach is more accurate than the other two; however it still 
makes enough mistakes to cause the heavier-tailed distribution of 
$\vec{y}^t$ as previously discussed. 

\subsection{Facebook Wall Posts}

I now test the proposed SBTM inference algorithm on a real data set, 
namely a dynamic social network of Facebook wall posts \citep{Viswanath2009}. 
Similar to the analysis by \citet{Viswanath2009}, I use $90$-day time 
steps from the start of the data trace in June 2006, with the 
final complete $90$-day interval ending in November 2008, resulting in $9$ 
total time steps.
I filter out people who were active for less than $7$ of the $9$ times as 
well as those with in- or out-degree less than $30$, leaving 
$462$ remaining people (nodes). 

\begin{figure}[tp]
\centering
\subfloat[$t=1$] 
	{\label{fig:FacebookAdj1}
	\includegraphics[width=2.2in]{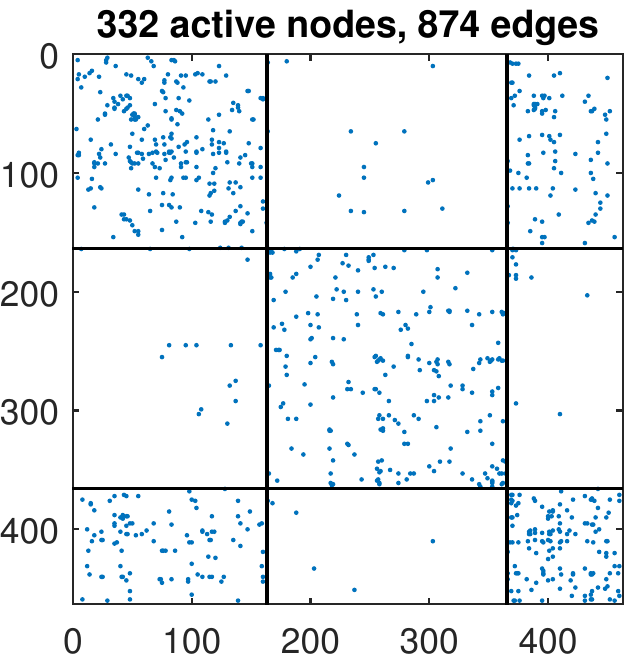}}
\qquad
\subfloat[$t=4$] 
	{\label{fig:FacebookAdj4}
	\includegraphics[width=2.2in]{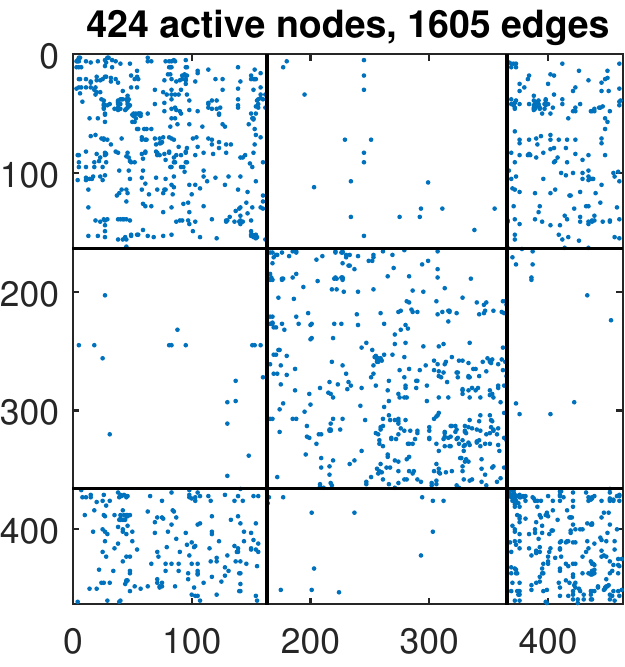}}
\qquad
\subfloat[$t=7$] 
	{\label{fig:FacebookAdj7}
	\includegraphics[width=2.2in]{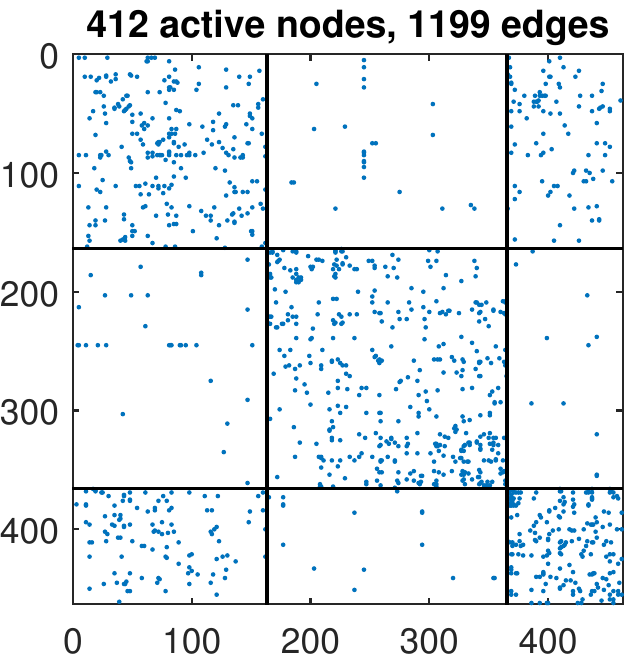}}
\qquad
\subfloat[$t=9$] 
	{\label{fig:FacebookAdj9}
	\includegraphics[width=2.2in]{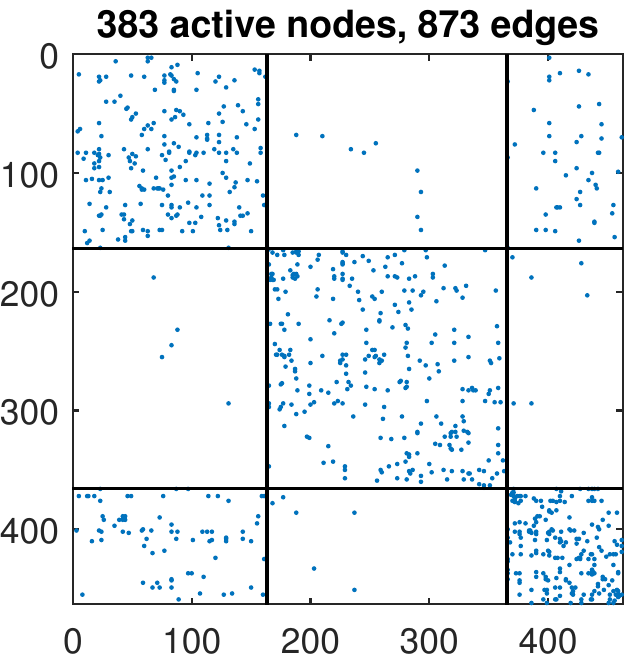}}
\caption{Adjacency matrices at four different time steps constructed from 
Facebook wall posts. 
The estimated classes at the final time $t=9$ are overlaid onto the adjacency 
matrices.}
\label{fig:FacebookAdj}
\end{figure}

I fit the SBTM to this dynamic network using Algorithm \ref{alg:SbtmInfer}, 
beginning with a spectral clustering initialization at the first time step. 
From examination of the singular values of the first snapshot, I choose a 
fit with $k=3$ classes. 
Visualizations of the class structure overlaid onto the adjacency matrices at 
several time steps are shown in Figure \ref{fig:FacebookAdj}. 
Notice that all of the classes are actually communities, with denser diagonal 
blocks compared to off-diagonal blocks. 
The initial snapshot contains only $332$ active nodes, so many new nodes 
enter the network over time. 
The networks are quite sparse, with the densest block having estimated 
marginal edge probability of about $0.08$. 
I find that the estimated probabilities of forming new edges is very low, 
less than $0.03$ over all time steps regardless of block. 
The probabilities of existing edges re-occurring show greater variation 
between blocks, ranging from about $0.18$ to $0.90$. 

\begin{figure}[tp]
\centering
\subfloat[Observed network] 
	{\label{fig:FacebookDurHistObs}
	\includegraphics[width=2.3in]{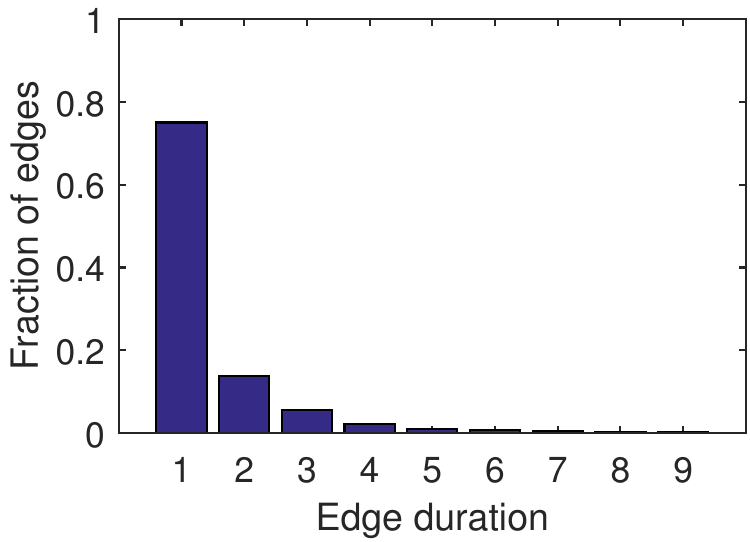}}
\quad
\subfloat[HM-SBM simulated networks] 
	{\label{fig:FacebookDurHistHmSbm}
	\includegraphics[width=2.3in]{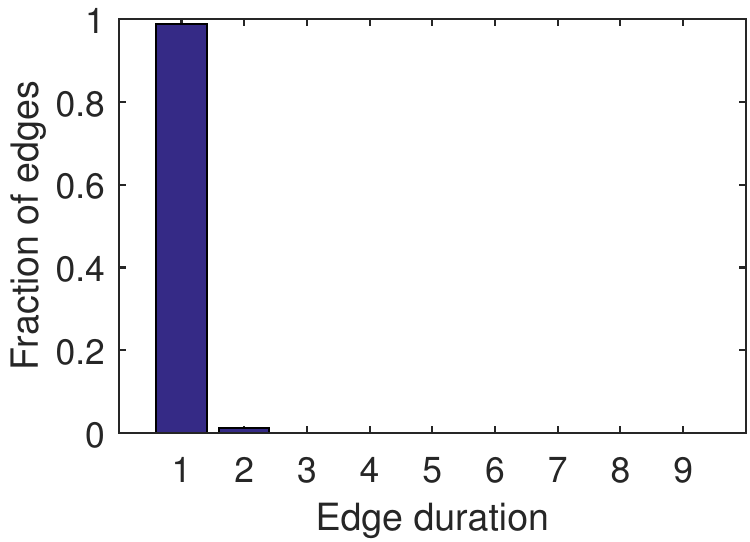}}
\quad
\subfloat[SBTM simulated networks] 
	{\label{fig:FacebookDurHistSbtm}
	\includegraphics[width=2.3in]{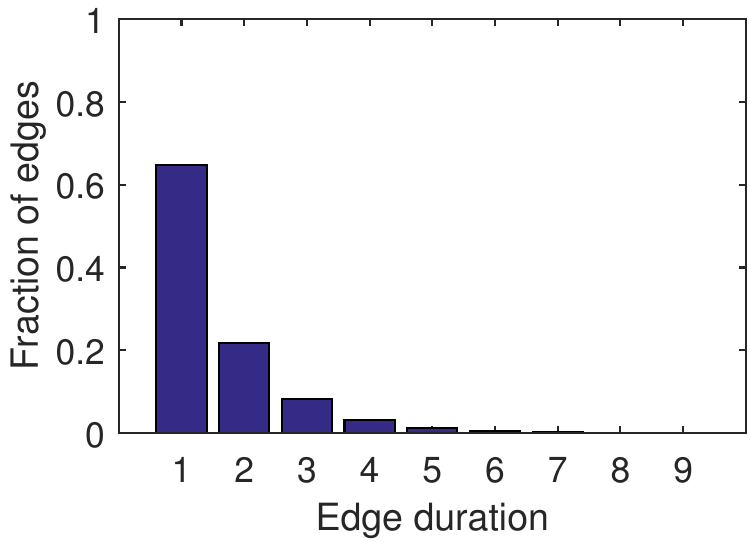}}
\caption[Histograms of edge durations in Facebook network]
{Histograms of edge durations in \subref{fig:FacebookDurHistObs} observed 
Facebook network, \subref{fig:FacebookDurHistHmSbm} 
simulated networks from HM-SBM fit to observed network, and
\subref{fig:FacebookDurHistSbtm} simulated networks from 
SBTM fit to observed network. 
The HM-SBM cannot reproduce the observed edge durations, unlike the SBTM.}
\label{fig:FacebookDurHist}
\end{figure}

A histogram of the edge durations observed in the network is shown in Figure 
\ref{fig:FacebookDurHistObs}. 
Notice that, despite the low densities of the blocks, more than $20\%$ of the 
edges appear over multiple time steps. 
I generate $10$ synthetic networks each from the HM-SBM and SBTM fits to the 
observed networks. 
The histogram of edge durations from synthetic networks generated from the 
HM-SBM is shown in Figure \ref{fig:FacebookDurHistHmSbm}. 
Due to the hidden Markov assumption, only the densities of the blocks are 
being replicated over time, and as such, the majority of edges are not 
repeated at the following time step. 
Compare this to the edge durations generated from the proposed SBTM, shown 
in Figure \ref{fig:FacebookDurHistSbtm}. 
Notice that a significant fraction of edges are indeed repeated in these 
synthetic networks, much like in the observed networks. 
These edge durations cannot be replicated by the HM-SBM. 
Thus the proposed SBTM provides better fits to the sequence of observed 
adjacency matrices and allows it to better forecast future interactions.

Notice also that the edge durations from the synthetic networks are actually 
slightly longer 
than from the observed networks. 
This is an artifact that appears because not all nodes are 
active at all time steps in the observed networks, causing edge durations 
to be shortened in the observed networks. 
One could perhaps replicate this effect by adding a layer to 
the dynamic model simulating nodes entering and leaving the network over time, 
which would be an interesting direction for future work. 

The proposed SBTM can also be extended to have edges depend directly on 
whether edges were present further back than just the previous time step. 
Such an approach would likely improve forecasting ability; 
however, it also increases the number of states that need to be estimated, 
which creates additional challenges that would make for interesting future 
work.

\subsubsection*{Acknowledgements}
The author thanks Prof.~Alan Mislove for providing access to the Facebook 
data analyzed in this paper.

\bibliographystyle{abbrvnat}
\bibliography{library_s}

\end{document}